\documentclass[conference,english]{IEEEtran}%
\usepackage[fleqn]{amsmath}
\usepackage{graphicx}
\usepackage{amssymb}
\usepackage{amsfonts}
\usepackage{amscd}
\usepackage{color}
\usepackage{babel}
\usepackage{subfigure}
\usepackage{algorithm}
\usepackage{algorithmicx}
\usepackage{listings}
\usepackage{epstopdf}
\usepackage{epsfig}
\usepackage{algpseudocode}
\usepackage{array}
\usepackage{multicol}
\usepackage{psfrag}
\usepackage{amsthm}
\usepackage{etoolbox}
\makeatletter
\patchcmd{\@makecaption}
  {\scshape}
  {}
  {}
  {}
\makeatother
\newtheorem{theorem}{Theorem}
\newtheorem{lemma}[theorem]{Lemma}

\begin{document}
\IEEEoverridecommandlockouts
\title{Low Power Analog-to-Digital Conversion in Millimeter Wave Systems: Impact of Resolution and Bandwidth on Performance}
\author{\IEEEauthorblockN{Oner Orhan, Elza Erkip, and Sundeep Rangan\thanks{
This work was supported in part by NSF CNS-1302336 and NYU WIRELESS.}}
\IEEEauthorblockA{NYU Polytechnic School of Engineering, Brooklyn, NY\\
Email: \text{onerorhan@nyu.edu}, \text{elza@nyu.edu}, \text{srangan@nyu.edu}}}
\maketitle

\begin{abstract}
The wide bandwidth and large number of antennas used in millimeter wave systems put a heavy burden on the power consumption at the receiver. In this paper, using an additive quantization noise model, the effect of
analog-digital conversion (ADC) resolution and bandwidth on the achievable rate is investigated for a multi-antenna system under a receiver power constraint.
Two receiver architectures, analog and digital combining, are compared in terms of performance.
Results demonstrate that:
(i)  For both analog and digital combining,
there is a maximum bandwidth beyond which the achievable rate decreases; (ii) Depending on the operating regime of the system, analog combiner may have higher rate but digital combining uses less bandwidth when only ADC power consumption is considered, (iii) digital combining may have higher rate when power consumption of all the components in the receiver front-end are taken into account.
\end{abstract}

\section{Introduction}\label{intro}
The millimeter wave (mmWave) bands, roughly between 30 and 300~GHz,
are an attractive candidate for next-generation cellular systems
due to the vast quantities of spectrum
and the potential for exploiting very high-dimensional
antenna arrays \cite{KhanPi:11-CommMag}-\cite{Rappaport2014-mmwbook}.
However, a significant issue in  realizing these systems
is power consumption, particularly in handheld mobile devices.
In addition to the high power consumed baseband
processing, one critical concern is the power consumption
in the analog-digital conversion (ADC) due to the need to process
large number of antenna outputs and very wide bandwidths.

Classical information theoretic formulation \cite{cover}
characterizes the maximum achievable communication rate on a channel
as a function on the signal-to-noise ratio (SNR) and bandwidth.
The broad goal of this paper is to understand what the information
theoretic limits on communication in the presence of constraints
on the ADC power consumption are
and how we can  design communication systems to meet these limits.
Although a closed form expression for the optimal input and capacity of
a 1-bit quantizer is given in \cite{madhow}, it is difficult to obtain the capacity and input distributions of multi-bit quantizer. Therefore, we
restrict our analysis to an additive quantization noise model (AQNM)
\cite{gray}.  Under this assumption, ADC power constraints can be
 easily abstracted as constraints on the sampling rate and quantizer noise
 level.
We can then obtain lower bounds on the capacity and
investigate the optimal bandwidth and resolution of the
ADC such that the
achievable rate is maximized.

We first study the effect of bandwidth and sampling rate of ADC on the performance of single-input and single output (SISO) system under the total receiver power budget including the ADC power cost.  We next consider multiple-input multiple-output (MIMO) point-to-point systems operating over a large bandwidth.
Our analysis considers both spatial multiplexing and beamforming
-- the two basic methods for MIMO systems \cite{TseV:07}.
Spatial multiplexing is favorable when the channel is bandwidth-limited and the SNR is high. However, beamforming provides power gain and is beneficial if the channel is power-limited as in most of the mmWave systems due to available wide bandwidth and high path loss \cite{sun}.

For mmWave MIMO systems, there are three receiver combining methods
that need to be considered: analog combining, digital combining and hybrid analog/digital combining \cite{sparse}.
Digital combining
is the method in use in conventional cellular transceivers
today where each antenna element has a separate pair of ADCs.  This method offers
the greatest flexibility, but highest power consumption for a given ADC
resolution and sampling rate.  An alternate design is to combine
the signals in analog (either in RF or IF) so that only one pair of ADCs is required
per stream
\cite{KohReb:07}-\cite{Heath:partialBF}.
Hybrid beamforming \cite{Heath:partialBF} uses a two stage combination of these designs.  See \cite{sun} for a general discussion.

Since power dissipation of ADC scales linearly in sampling rate and exponentially in the number of bits per sample \cite{lee}, it may not be desirable to operate the system over the full bandwidth and high resolution. Therefore, digital combining can save energy by using narrow band (low sampling rate) and low resolution ADCs while increasing its effective received SNR and allowing the spatial multiplexing at the transmitter for better performance. As a result, analysis of performance trade-off of analog and digital combining is essential for power limited receivers.

\begin{figure}
\centering
\subfigure[]{
\includegraphics[scale=0.5,trim= 20 10 10 0]{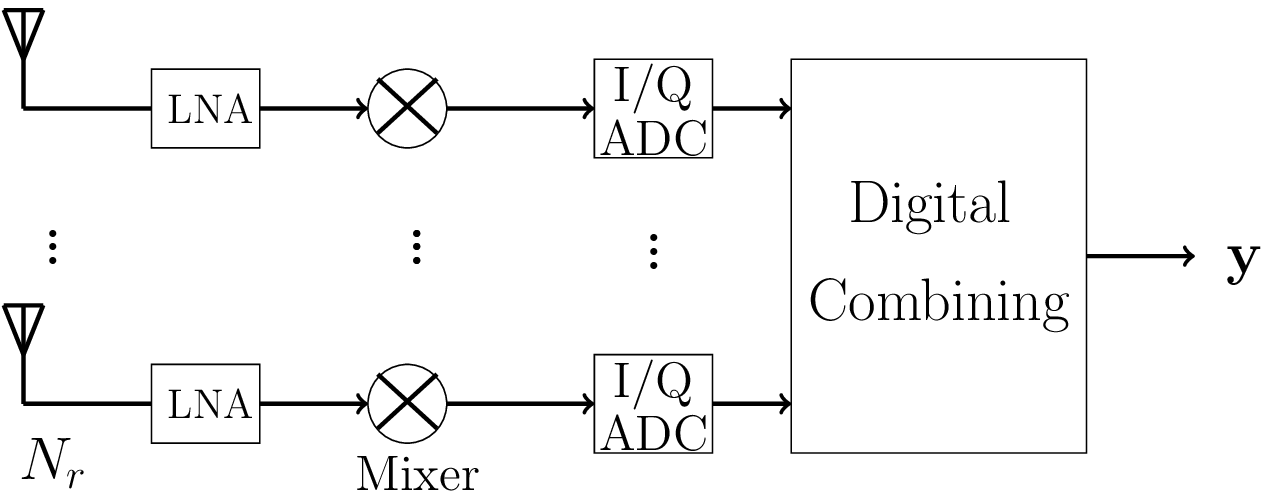}
\label{sys:subfig1}
}
\subfigure[]{
\includegraphics[scale=0.52,trim= 20 10 10 0]{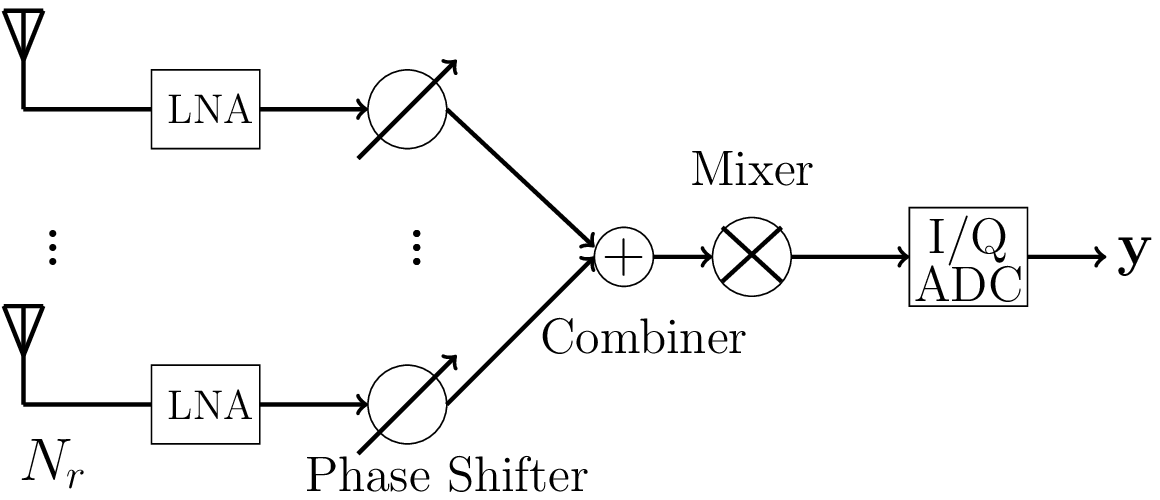}
\label{sys:subfig2}
}
\caption{Receiver architectures: (a) Digital combining (b) Analog combining.}
\label{sys}%
\vspace{-0.1in}
\end{figure}

To compare these architectures,
we derive achievable rates of the MIMO system with both
analog combining and digital combining and study the effect of bandwidth, ADC resolution and the number of antennas on the achievable rate. Constraints
are placed on the total front-end receiver power or only
the ADC power. We expect most of the simulation results for the hybrid analog/digital combining to be somewhere between analog and digital combining, and hence we leave detailed analysis of this as a future work.

Under these power constraints, our results show that:
\begin{itemize}
\item For a point-to-point SISO system when bandwidth is large as in mmWave systems, there exists an optimal bandwidth and resolution of the ADC.
    As a consequence, it may not be desirable to operate the system over the full bandwidth.
\item For a point-to-point MIMO system when only ADC power consumption is considered, we see that at least in some operating regimes, analog combining may have higher maximum achievable rate than digital
    combining.  However, in these cases, digital combining can attain its maximum rate at much less bandwidth. Indeed, as the number of antennas at the transmitter and receiver $N$ increases, the utilized bandwidth decreases by a factor of $\frac{1}{N}$.
\item For a point-to-point MIMO system when the power consumption of all receiver components including ADCs are considered, digital combining may be able to achieve a higher rate when the channel state information is available at the transmitter.
\end{itemize}

\subsection*{Previous work}

In recent years, energy efficient transceiver architectures such as use of low resolution ADCs and hybrid analog/digital precoding (combining) has attracted significant interest. The limits of communication over additive white Gaussian channel with low resolution (1-3) bits ADCs at the receiver is studied in \cite{madhow}. This is extended in \cite{fettweis} to the capacity of the 1-bit ADC with oversampling at the receiver. The bounds on the capacity of the MIMO channel with 1-bit ADC at high and low SNR regimes are derived in \cite{heath} and \cite{mezgani}, respectively. Using AQNM of ADC, the joint optimization of ADC resolution with the number of antennas in a MIMO channel is studied in \cite{nossek}. There is also rich literature on the hybrid analog/digital transceivers \cite{sparse}-\cite{robert}. While \cite{sparse} provides efficient hybrid precoding and combining algorithms for sparse mmWave channels which performs close to full digital solution, \cite{robert} combines efficient channel estimation with the hybrid precoding and combining algorithm in \cite{sparse}. However, to the best of our knowledge, there is no work which considers the effects of both sampling rate (bandwidth) and resolution of ADC on the performance of the system under a total receiver power constraint. In addition, we investigate and compare the performance of analog and digital combining and optimize the resolution, bandwidth and the number of antennas.

It should be noted that our analysis would likely apply most closely to data
plane traffic where the overhead for channel tracking is relatively small.
For initial synchronization, where the channel or even the presence
of the transmitting base station may not be known, digital combining
can have distinct advantages not accounted for in this analysis
-- see \cite{BarHosCellSearch:14-spawc}.

\section{System Model}\label{sys model}
We consider a point-to-point mmWave communication system operating over a bandwidth $W_{tot}$ Hz. We assume that there are $N_t$ antennas at the transmitter and $N_r$ antennas at the receiver. We assume additive white Gaussian noise (AWGN) with power spectral density $\frac{N_0}{2}$ Watts/Hz. The transmitted signal has average power constraint $P$ Watts. The channel exhibits frequency selective fading, which is independent across frequency bands and across antennas. The instantaneous fading realizations are assumed to be known at the receiver.

We consider two receiver architectures as shown in Figure~\ref{sys}:
digital combining and analog combining. As shown in the figure,
for digital combining,
 ADCs are employed to quantize the signal before the baseband combiner.
 Each block labeled ``I/Q ADC" represents two ADCs -- one for the inphase and another for the quadrature components, both
with sampling rate equal to the Nyquist rate.
In the analog combining architecture, the signals are combined in analog
with phase shifters, and then digitized with a single I/Q ADC.
Thus, for RF analog combining,
we have $N_r$ low noise amplifiers (LNAs),
$N_r$ phase shifters, one combiner, one mixer and one I/Q ADC,
while the digital combining architecture consists of only $N_r$ LNAs and $N_r$ mixers \cite{Taghivand}, but $N_r$ ADCs.

We assume that for either analog or digital combing,
each ADC consists of a $b$-bin scalar quantizer. Note that in \cite{madhow}, it is proved that for $b$-bin output quantization, at most $b+1$ mass points at the input are enough to achieve the capacity. However, obtaining the optimal input distribution for arbitrary number of quantization bins is difficult. Therefore, in this paper, in order to get insights, we use an AQNM for the quantizer and find a lower bound to the capacity of the corresponding channel by assuming Gaussian quantization noise and Gaussian inputs. Further benefits of AQNM include ease of implementation since the decoder can use standard linear processing and Gaussian decoding.

\subsection{Additive Quantization Noise Model (AQNM)}
We denote the output of the ADC corresponding to input $z$ by $Q(z)$. We consider that the quantizer output $z_q=Q(z)$ is chosen such that $z_q=E[z|z_q]$. The quantizer $Q(\cdot)$ can be represented by the following AQNM \cite{rangan}:
\begin{eqnarray}\label{eq0}
z_q &=&\alpha z+n_q,
\end{eqnarray}
where $n_q$ is the additive quantization noise such that $z$ and $n_q$ are uncorrelated. Note that
\begin{eqnarray}\label{eq1}
E[n_q]&=&(1- \alpha) E[z]\\\label{eq2}
\sigma_{n_q}^2&=&(1-\alpha)\alpha \sigma_z^2,
\end{eqnarray}
Here, $\sigma_{n_q}^2$ is variance of additive quantization noise. Accordingly, $\alpha$ can be computed as:
\begin{eqnarray}\label{eq3}
\alpha =1-\beta,
\end{eqnarray}
where $\beta=\frac{\sigma_{e_q}^2}{\sigma_z^2}$ where $\sigma_{e_q}^2$ is the variance of quantization error, $e_q=z-z_q$, and $\sigma_z^2$ is variance of the quantization input. In \cite{gray}, $\frac{1}{\beta}$ is referred as the coding gain.


 Motivated by \cite{cover} we assume that $\mathbf{n_q}$ has the Gaussian distribution. Furthermore, since the Gaussian input maximizes the mutual information between the input and the output when the noise is Gaussian, we assume that the input to the MIMO channel is jointly Gaussian. Accordingly, for non-uniform scalar MMSE quantizer of a Gaussian random variable, $\beta$ can be approximated to $\beta=\frac{\pi \sqrt{3}}{2}b^{-2}$, where $b$ is the number of quantization bins \cite{gray}. In this paper, without loss of generality we assume that $\beta=ab^{-2}\leq 1$ for some constant $a>0$. Note that as $b \rightarrow \infty$, $\beta \rightarrow 0$.

\subsection{Power Consumption of the Receiver}
We consider two scenarios to account for the power consumption at the receiver. In the first one,
we assume a total processing power budget. Denoting the power consumption of the LNA, phase shifter, combiner, mixer and ADC  by $P_{LNA}$, $P_{PS}$, $P_C$, $P_M$ and $P_{ADC}$, respectively, the total power consumption $P_{tot}$ of analog and digital combining in terms of $N_r$ are respectively given by
\begin{eqnarray}\label{pow1}
P_{tot}=N_r (P_{LNA}+P_{PS})+P_C+ P_M +2 P_{ADC},
\end{eqnarray}
and
\begin{eqnarray}\label{pow2}
P_{tot}=N_r (P_{LNA}+P_{M}+2 P_{ADC}).
\end{eqnarray}

We assume that the power consumption of the LNA, phase shifter, combiner, and mixer are independent of the bandwidth. The power dissipation of ADC scales linearly in sampling rate and exponentially in the number of bits per sample \cite{lee}. Assuming sampling at the Nyquist rate, the figure of merit of ADC power consumption is modeled as
\begin{eqnarray}\label{adc_power}
P_{ADC}=c W b,
\end{eqnarray}
where $W$ is sampling rate (bandwidth), $b$ is the number of quantization bins of ADC, and $c$ is the energy consumption per conversion step, e.g. $c=494$ fJ \cite{Murmann}.

In the second scenario, we only put a constraint on the total power used by all the ADCs in the receiver. While this helps avoid comparing power consumption of  different components (which are designed using possibly different technologies), the analysis may only give a partial indication of what could happen in practice.

\section{Accuracy of AQNM}\label{accuray}
In order to illustrate the accuracy of the AQNM, we consider a simple SISO system, i.e., $N_t=N_r=1$, with fixed complex channel gain $h$. We compare the rate achieved by Gaussian inputs in the AQNM  with the capacity computed for the 2, 4, and 8-bin quantizers \cite{madhow}.  Using the AQNM in (\ref{eq0}), the equivalent received signal after the quantizer can be written as
\begin{eqnarray}
y=(1-ab^{-2}) h x +(1-ab^{-2}) n + n_q,
\end{eqnarray}
where $\sigma_{n_q}^2=ab^{-2}(1-ab^{-2})(|h|^2 P + N_0)$. Then, the following rate is achievable when operating over a bandwidth $W$:
\begin{eqnarray}\label{ch1}
R=W\log_2\left(1+\frac{(1-ab^{-2}) |h|^2 P}{ab^{-2} |h|^2 P +WN_0}\right) \quad \text{bits/sec}.
\end{eqnarray}
Note that the rate given in (\ref{ch1}) is monotonically increasing and concave function of $P$ for a given bandwidth and resolution.

As shown in Figure 2, the rate in (\ref{ch1}) lower bounds the capacity. At low SNRs the gap between the capacity and the rate under the AQNM is small, and as the quantization resolution increases the gap decreases. For example, at $-10$dB SNR, the rate assuming AQNM is $96\%$ and $99\%$ of the capacity for $2$-bin and $8$-bin quantization, respectively. However, at high SNR, the gap between the capacity and the rate of the AQNM increases. For example, at $20$dB SNR, the rate under the AQNM is $72\%$ and $77\%$ of the capacity for $2$-bin and $8$-bin quantization, respectively. In this regime, the rate achieved by an equiprobable $b$ point input distribution is close to the capacity which can be approximated by $W\log_2(b)$ bits/s \cite{madhow}. Similarly, in the high SNR regime the rate in (\ref{ch1}) can be approximated by $W\log_2(\frac{b}{a})/2$, or $-W\log_2(\beta)/2$, bits/s. Therefore, the gap between the capacity and the rate under the AQNM becomes $W\log_2(\sqrt{a})$, or $W\log_2(b\sqrt{\beta})$ as SNR increases. However, based on the fact that the capacity and the achievable rate under the AQNM show similar trends, we believe that the conclusions obtained in this paper would be valid when optimal input distribution is used, even though the actual numerical values of optimal bandwidth, resolution etc. may be slightly different.

\begin{figure}[t]\label{sim1}
\centering
\includegraphics[scale=0.65,trim= 0 0 0 0]{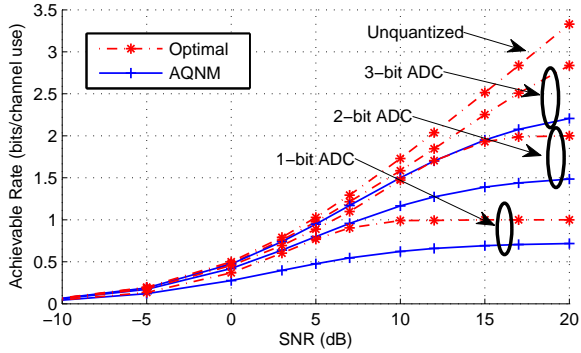}
\caption{Achievable rate versus SNR for the optimal input distribution and AQNM with Gaussian input. $\beta=0.363, 0.118, 0.037$ for $[2, 4, 8]$-bin quantizers.}
\end{figure}

\section{The Optimal Bandwidth and Resolution of ADC for SISO Systems}\label{SISO}
In order to get insights on the bandwidth usage and optimal resolution of ADCs, in this section we limit our investigation to the SISO channel. As in Section~\ref{accuray}, we assume the channel gain is constant and we investigate how the achievable rate in (\ref{ch1}) can be maximized by optimally choosing  the bandwidth $W$ and the number of quantization bins $b$ under a total receiver power $P_{tot}$ constraint. For $P_{tot}$, we follow the model in (\ref{pow2}). The optimal bandwidth and resolution can be computed by the following optimization problem.
\begin{subequations}\label{ppp:4}
\begin{eqnarray}\label{ppp:4a}
\underset{b, W}{\operatorname{max}} && W\log_2\left(1+\frac{\left(1-a b^{-2}\right) |h|^2 P}{a b^{-2} |h|^2 P +W N_0}\right) \\\label{ppp:4b}
\text{s.t.}~ && P_{LNA}+ P_{M} + 2cWb \leq P_{tot},  \\ \label{ppp:4c}
&& 0\leq W \leq W_{tot}, \quad 0 \leq b.
\end{eqnarray}
\end{subequations}

\begin{lemma}\label{lemma1}
For the optimization problem in (\ref{ppp:4}) when $W_{tot}$ is large, there exists an optimal bandwidth beyond which the achievable rate in (\ref{ppp:4a}) starts decreasing.
\end{lemma}
\begin{proof}
Since the rate in (\ref{ch1}) is monotonically increasing function of $b$, the constraint in (\ref{ppp:4b}) must be satisfied with equality. Therefore, from (\ref{ppp:4b}) we can argue that $b$ is equal to $\frac{P_{tot}-P_{LNA}-P_{M}}{2Wc}$. Accordingly, the optimization problem in (\ref{ppp:4}) can be reformulated as follows:
\begin{eqnarray}\label{eq5a}
\underset{W}{\operatorname{max}} && W \log_2\left(\frac{|h|^2P + W N_0}{\frac{a 2^2 c^2 W^2 |h|^2P}{(P_{tot}-P_{LNA}-P_{M})^2}+WN_0}\right),
\end{eqnarray}
where $0 \leq W \leq W_{tot}$. Here we assume $W_{tot} \rightarrow \infty$. The objective function in (\ref{eq5a}) is a concave function of $W$ when $P_{tot}-P_{LNA}-P_{M}>0$, $W>0$ and $a>0$. We can argue that as $W\rightarrow 0$, the objective function in (\ref{eq5a}) goes to 0, as $W\rightarrow \infty$ the objective function in (\ref{eq5a}) goes to $-\infty$. In addition, the objective function in (\ref{eq5a})  is positive when $\frac{a 2^2 c^2 W^2}{(P_{tot}-P_{LNA}-P_{M})^2}=ab^{-2}<1$. Therefore we can conclude that when $ab^{-2}< 1$, there exists an optimal finite bandwidth $W^*$ for which rate in (\ref{ppp:4}) is maximized.
\end{proof}

Lemma \ref{lemma1} shows that while for AWGN channels (even for those whose outputs are quantized)  the capacity increases as the bandwidth gets larger, under receiver power constraints, this is no longer the case. This suggests that even though the bandwidth is abundant in mmWave systems, since ADC is power hungry, it may not be desirable to operate the system over its full bandwidth. Similarly, there exits an optimal number of quantization levels that maximizes capacity under receiver power constraints.

In order to show the effect of SNR on the optimal bandwidth and  the number of quantization bins, we  set $P_{tot}-P_{LNA}-P_{M}=2P_{ADC}= 20$ mW and  the energy consumption per conversion step $c=494$ fJ \cite{Murmann}. The total bandwidth is $W_{tot}= 7$ GHz \cite{daniel}. We consider scalar non-uniform MMSE quantizer at the ADC and compute $\beta$ using Lloyd-Max algorithm \cite{max}. As shown in Figure 3, in the low SNR regime it is optimal to use a lower bandwidth and larger number of quantization bins. This allows the transmitter to spread the available power over a smaller band and increase effective SNR. Here, SNR refers to the received SNR at the full bandwidth $W_{tot}$, i.e., $|h|^2{P}/{W_{tot}N_0}$. Note that beyond 5dB SNR, the optimal resolution and bandwidth are $b=3$ and $W=6.75$ GHz, which is slightly less than the total available bandwidth 7 Ghz.
\begin{figure}[ht]\label{sim4}
\centering
\includegraphics[scale=0.65,trim= 0 0 0 0]{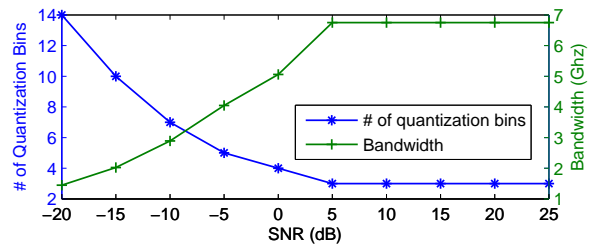}
\caption{The optimal bandwidth and number of quantization bins versus received SNR at the full bandwidth for the SISO channel. $W_{tot}=7$ Ghz. Total ADC power consumption is fixed at $20$mW and $c=494$ fJ.}
\end{figure}


\section{MIMO System}\label{MIMO}
In this section, we investigate the achievable rate for the MIMO system with digital and analog combining. For both receiver architectures we consider two scenarios: i) Perfect channel state information at the transmitter (CSIT) where instantaneous channel realizations are known at the transmitter; ii) no CSIT,  where the transmitter only knows the channel state.

\subsection{Digital Combining}\label{MIMO_dig}
Using AQNM, we can obtain the equivalent channel per frequency band as follows:
\begin{eqnarray}\label{eq20}
\mathbf{y_q} =(1-a b^{-2})\mathbf{H}\mathbf{x}+(1-a b^{-2})\mathbf{n} +\mathbf{n_q},
\end{eqnarray}
where $R_{n_qn_q}=a b^{-2}(1-a b^{-2})\text{diag}(\mathbf{H} R_{xx}\mathbf{H^H}+W R_{nn})$. Here $R_{nn}$ is the noise covariance matrix and $R_{xx}$ is the input covariance matrix. Then, the following rate is achievable with digital combining operating over a bandwidth $W \leq W_{tot}$:
\begin{equation}\label{eq17}
E_H\left[W\log_2\left|\mathbf{I}+\frac{(1-a b^{-2})\mathbf{H}R_{xx}\mathbf{H^H}}{a b^{-2}\text{diag}(\mathbf{H} R_{xx}\mathbf{H^H})+W R_{nn}}\right|\right].
\end{equation}

Since we consider independent identically distributed fading across antennas and frequency bands,  the optimal input covariance matrix under no CSIT is $R_{xx}=\frac{P}{N_t}\mathbf{I}$ \cite{telatar}.

Under CSIT, since the input covariance matrix appears as noise in (\ref{eq17}), using singular value decomposition (SVD) of the channel matrix for $R_{xx}$ does not immediately convert the MIMO channel into multiple parallel channels. However, in this section for ease of implementation we use the SVD of the channel matrix  to determine $R_{xx}$ which provides a further lower bound to the achievable rate. We also assume constant power allocation across frequency bands. Then, denoting SVD of $\mathbf{H}$ by $\mathbf{U}\Lambda \mathbf{V^H}$, we can rewrite (\ref{eq20}) as
\begin{eqnarray}
\mathbf{y_q}&=&(1-a b^{-2})\mathbf{U}\Lambda \mathbf{V^H}\mathbf{x}+(1-a b^{-2})\mathbf{n} +\mathbf{n_q},
\end{eqnarray}
Using the transmit beamforming matrix $\mathbf{V}$ and the digital combining matrix $\mathbf{U^H}$, we obtain the received signal as
\begin{eqnarray}
\hspace{-0.05in}\mathbf{y}& =  & \mathbf{U^H}\mathbf{y_q} \\
&=&(1-a b^{-2})(\mathbf{U^H}\mathbf{U}\Lambda \mathbf{V^H}\mathbf{V}\mathbf{\tilde{x}}+\mathbf{U^H}\mathbf{n}) +\mathbf{U^H}\mathbf{n_q}\\\label{eq21}
&=&(1-a b^{-2})\Lambda \mathbf{\tilde{{x}}}+(1-a b^{-2})\mathbf{\tilde{n}} +\mathbf{\tilde{n}_q},
\end{eqnarray}
where $\mathbf{\tilde{x}}=\mathbf{V^H}\mathbf{x}$, $\mathbf{\tilde{n}}=\mathbf{U^H}\mathbf{n}$ and $\mathbf{\tilde{n}_q}=\mathbf{U^H}\mathbf{n_q}$. Here $\mathbf{\tilde{n}}$ has the same distribution as $\mathbf{n}$. Allocating power across the eigenvalues of  $\Lambda$  according to  the waterfilling solution, we obtain the following achievable rate:
\begin{equation}\label{eq11}
E_H\left[W\log_2\frac{|R_{n'n'}+(1-a b^{-2})\Lambda \mathbf{Q}\Lambda |}{|R_{n'n'}|}\right],
\end{equation}
where $R_{n'n'}=a b^{-2}\mathbf{U^H}\text{diag}(\mathbf{U}\Lambda \mathbf{Q} \Lambda \mathbf{U^H})\mathbf{U}+W R_{nn}$. Here $\mathbf{Q}$ denotes the diagonal matrix that contains the power levels. Note that the above waterfilling strategy would be optimal in the case of no quantization, where the power dependent quantization noise term is zero.


\subsection{Analog Combining:}
Using AQNM, we can obtain the equivalent channel of analog combining per frequency band as follows:
\begin{eqnarray}\label{eq13}
\mathbf{y_q}  = (1-a b^{-2})\mathbf{w_r^H}\mathbf{H}\mathbf{w_t}x+(1-a b^{-2})\mathbf{w_r^H}\mathbf{n} +\mathbf{n_q},
\end{eqnarray}
where $\sigma_{n_q}^2=a b^{-2}(1-a b^{-2})(|\mathbf{w_r^H}\mathbf{H}\mathbf{w_t}|^2P+N_r N_0)$. Here, $\mathbf{w_r}$ is the analog combining vector such that $|w_{r,i}|=1$ $i=1, \ldots,  N_r$, and $\mathbf{w_t}\in \mathbb{C}^{N_t \times 1}$ is the digital beamforming vector at the transmitter. Then, the achievable rate of analog combining operating over a bandwidth $W \leq W_{tot}$ is given by:
\begin{equation}\label{eq19}
E_H\left[\underset{\mathbf{w_r},\mathbf{w_t}}{\operatorname{max}}  W\log_2\left(1+\frac{(1-a b^{-2})P|\mathbf{w_r^H}\mathbf{H}\mathbf{w_t}|^2}{a b^{-2}P|\mathbf{w_r^H}\mathbf{H}\mathbf{w_t}|^2+W N_r N_0}\right)\right].
\end{equation}
Note that the above expression is monotonically increasing function of the received power $P|\mathbf{w_r^H}\mathbf{H}\mathbf{w_t}|^2$. Therefore, the achievable rate is maximized when the term $|\mathbf{w_r^H}\mathbf{H}\mathbf{w_t}|^2$ is maximized. The maximum value that can be obtained depends of the availability of CSIT.

\subsubsection{ Analog Combining without CSIT}
Since there is no CSIT and the channel is symmetric,  the optimal transmit beamforming vector is $w_{t,i}=\frac{1}{\sqrt{N_t}}$, $i=1,..., N_t$. We the have \cite{love}
\begin{eqnarray}
|\mathbf{w_r^H}\mathbf{H}\mathbf{w_t}|^2 &= & \left|\sum_{i=1}^{N_r} e^{j\phi_i}\frac{1}{\sqrt{N_t}}\sum_{j=1}^{N_t} h_{ij}\right|^2\\
&\leq & \frac{1}{N_t}\left(\sum_{i=1}^{N_r}\left|\sum_{j=1}^{N_t} h_{ij}\right|\right)^2,
\end{eqnarray}
where $\phi_i$ is phase of the $i$th element of $\mathbf{w_r}$. The inequality above holds with equality when  $\phi_i$ is  chosen as the phase of $\sum_{j=1}^{N_t} h_{ij}$ plus $\pi$.

\subsubsection{Analog Combining with CSIT}
At the transmitter digital beamforming in the form of maximum ratio transmission maximizes the received power. For a given analog combining vector $\mathbf{w_r}$ we have
\begin{eqnarray}
|\mathbf{w_r^H}\mathbf{H}\mathbf{w_t}|^2 \leq \|\mathbf{w_r^H}\mathbf{H}\|^2  \|\mathbf{w_t}\|^2.
\end{eqnarray}
The above inequality satisfied with equality when $\mathbf{w_t}=\frac{\mathbf{H^H}\mathbf{w_r}}{\|\mathbf{H^H}\mathbf{w_r}\|}$. Note that the normalization is due to the power constraint at the transmitter, where, as in Section~\ref{MIMO_dig} we assume no power allocation across frequency bands. Then, we have $|\mathbf{w_r^H}\mathbf{H}\mathbf{w_t}|^2=\|\mathbf{w_r^H}\mathbf{H}\|^2$.


\section{Illustration of the  Results}
In this section, we provide numerical results to show the effects of  ADC bandwidth and resolution,  the number of antennas, receiver architecture and power consumption on the achievable rate. We consider scalar non-uniform MMSE quantizer at each ADC and compute $\beta$ using Llyod-Max algorithm \cite{max}. We set the energy consumption per conversion step $c$ to $494$ fJ \cite{Murmann}. We consider that the maximum available bandwidth is $W_{tot}=7$ Ghz  for operation in the  57-64 Ghz band \cite{daniel}. Throughout SNR refers to SNR at the full bandwidth, that is, ${P}/{W_{tot}N_0}$. All results are obtained under independent Rayleigh fading across space and frequency.

We first consider only the total ADC power consumption which is set to 20 mW. Hence for analog combining $2P_{ADC}=20$ mW while for digital $2N_r P_{ADC}=20$ mW. We examine the effect of SNR on the rates of   $1 \times 3$ SIMO, and $ 3 \times 3$  MIMO systems with and without CSIT for analog and digital combining as shown in Figure 4. The figure includes achievable  rates for only the best resolution levels. Note that performances of the SIMO and the MIMO systems with no CSIT are the same for  analog combining. We observe that the optimal resolution is 3-bin quantization for all scenarios investigated except for  SIMO with digital combining for which 2-bin quantization optimal.  As shown in the figure,  analog combining achieves higher rates than  digital combining for all the cases studied even though the MIMO system with digital combining has the advantage of spatial multiplexing. The benefit of spatial multiplexing can be seen from the gap between the rates of the SIMO and MIMO systems with digital combining. However,  MIMO system with digital combining utilizes one-third ($1/N_r$) of the bandwidth used by the analog combining architecture. This is due to fact that the total ADC power budget in the case of digital combining is shared among  antennas lowering the bandwidth usage.  For the scenario investigated, optimal bandwidth usage for analog combining is $W=6.75$ GHz  while the usage for the digital combining architecture is  2.25 Ghz and 3.37 Ghz  for MIMO and SIMO respectively.

\begin{figure}[ht]\label{sim5}
\centering
\includegraphics[scale=0.65,trim= 0 0 0 0]{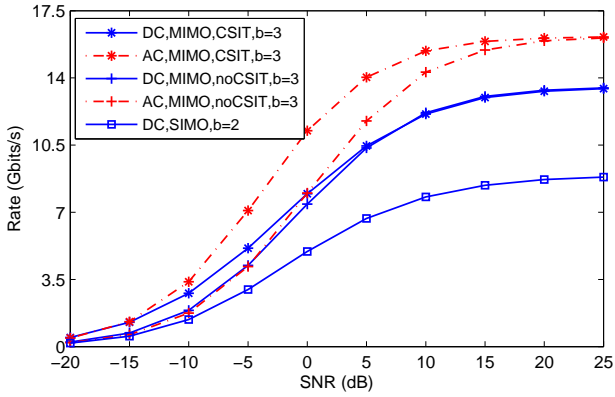}
\caption{Achievable rates for $1 \times 3$ SIMO and $3 \times 3$ MIMO systems  with analog (AC) and digital combining (DC) versus SNR, with $b$ referring to the optimal number of quantizer bins. Total ADC power consumption is fixed at $20$mW. $W_{tot}=7$ Ghz,  $c=494$ fJ.}
\end{figure}

\begin{table*}[t]\label{table}
\begin{tabular}{|l|l|l|l|l|l|l|l|l|l|l|l|l|l|l|l|l|}
\hline
                                                         & \multicolumn{4}{l|}{DC, MIMO, CSIT}                                                                                                                                       & \multicolumn{4}{l|}{AC, MIMO, CSIT}                                                                                                                                        & \multicolumn{4}{l|}{DC, MIMO, noCSIT}                                                                                                                                     & \multicolumn{4}{l|}{AC, MIMO,no CSIT}                                                                                                                                     \\ \hline
\begin{tabular}[c]{@{}l@{}}$P_{tot}$\\ (mW)\end{tabular} & \begin{tabular}[c]{@{}l@{}}Rate\\ (Gbits/s)\end{tabular} & $b$ & $N$ & \begin{tabular}[c]{@{}l@{}}W\\ (Ghz)\end{tabular} & \begin{tabular}[c]{@{}l@{}}Rate\\ (Gbits/s)\end{tabular} & $b$  & $N$ & \begin{tabular}[c]{@{}l@{}}W\\ (Ghz)\end{tabular} & \begin{tabular}[c]{@{}l@{}}Rate\\ (Gbits/s)\end{tabular} & $b$ & $N$ & \begin{tabular}[c]{@{}l@{}}W\\ (Ghz)\end{tabular} & \begin{tabular}[c]{@{}l@{}}Rate\\ (Gbits/s)\end{tabular} & $b$  & $N$ & \begin{tabular}[c]{@{}l@{}}W\\ (Ghz)\end{tabular} \\ \hline
100                                                      & 5.52                                                     & 5 & 1                                                      & 7                                                 & 2.16                                                     & 2  & 1                                                      & 1.73                                              & 5.40                                                     & 5 & 1                                                      & 7                                                 & 2.21                                                     & 2  & 1                                                     & 1.73                                              \\ \hline
150                                                      & 10.03                                                    & 2 & 2                                                      & 6.4                                               & 5.70                                                     & 6  & 1                                                      & 7                                                 & 7.20                                                     & 2 & 2                                                      & 6.4                                               & 5.64                                                     & 7  & 1                                                     & 6.9                                               \\ \hline
200                                                      & 13.09                                                    & 5 & 2                                                      & 7                                                 & 11.97                                                    & 6  & 2                                                      & 6.67                                              & 9.21                                                     & 5 & 2                                                      & 7                                                 & 8.40                                                     & 6  & 2                                                     & 6.67                                              \\ \hline
250                                                      & 17.85                                                    & 3 & 3                                                      & 6.88                                              & 14.91                                                    & 4  & 3                                                      & 7                                                 & 11.00                                                    & 3 & 3                                                      & 6.88                                              & 10.11                                                    & 4  & 3                                                     & 7                                                 \\ \hline
300                                                      & 20.51                                                    & 2 & 4                                                      & 6.4                                               & 17.68                                                    & 11 & 3                                                      & 7                                                 & 11.98                                                    & 5 & 3                                                      & 7                                                 & 11.58                                                    & 11 & 3                                                     & 7                                                 \\ \hline
350                                                      & 24.43                                                    & 4 & 4                                                      & 6.34                                              & 20.58                                                    & 10 & 4                                                      & 7                                                 & 13.23                                                    & 3 & 4                                                      & 7                                                 & 13.36                                                    & 10 & 4                                                     & 7                                                 \\ \hline
400                                                      & 28.07                                                    & 3 & 5                                                      & 6.05                                              & 22.47                                                    & 9  & 5                                                      & 7                                                 & 14.24                                                    & 5 & 4                                                      & 7                                                 & 14.85                                                    & 9  & 5                                                     & 7                                                 \\ \hline
450                                                      & 31.85                                                    & 4 & 5                                                      & 6.84                                              & 23.73                                                    & 16 & 5                                                      & 7                                                 & 15.40                                                    & 4 & 5                                                      & 6.84                                              & 15.89                                                    & 8  & 6                                                     & 6.96                                              \\ \hline
500                                                      & 36.40                                                    & 3 & 6                                                      & 6.88                                              & 25.62                                                    & 15 & 6                                                      & 7                                                 & 16.31                                                    & 3 & 6                                                      & 6.88                                              & 16.73                                                    & 15 & 7                                                     & 3.38                                              \\ \hline
\end{tabular}
\caption { {Performance comparison for  the $N \times N$ MIMO system with analog (AC) and digital (DC) combining, with and without CSIT. The table shows the achievable rate in Gbits/s,  the optimal resolution $b$, the optimal number of  antennas $N$, and the optimal bandwidth $W$ in Ghz under each scenario for a given total receiver power budget $P_{tot}$. $W_{tot}=7$ Ghz, $P_{LNA}=39$mW, $P_C=19.5$mW, $P_{PS}=19.5$mW, and $P_M=16.8$mW, $c=494$ fJ.}}
\end{table*}

Next,  we investigate the effect of the number of antennas on the achievable rates of the SIMO and the MIMO systems at 0dB SNR. The MIMO system has $N$ transmit and $N$ receive antennas, while for SIMO, $N_r=N$.  We assume total ADC power budget of $20$ mW as in Figure 4. Our results are shown in Figure 5, where we only include the best quantization resolution (3-bins for  most of the scenarios investigated, 2- and 3-bins for SIMO system with digital combining).  We observe that the rate of the MIMO system with analog and digital combining increases with the number of antennas. In addition, although analog combining has higher capacity for all the cases, the optimal bandwidth of  digital combining scales as  ${1}/{N}$ as opposed to the constant bandwidth usage of  analog combining.

\begin{figure}[ht]\label{sim6}
\centering
\includegraphics[scale=0.65,trim= 0 0 0 0]{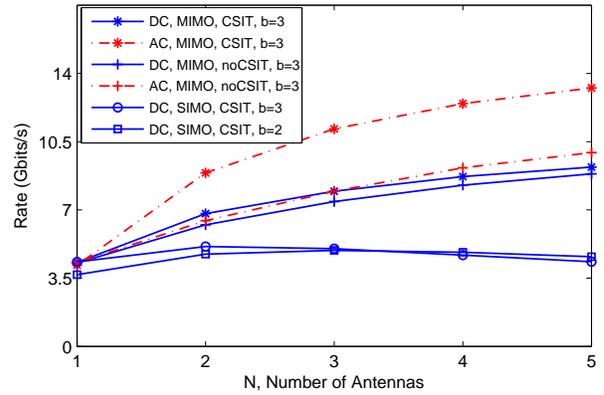}
\caption{Achievable rates for $1\times N$ SIMO, and $N \times N$ MIMO systems  with analog (AC) and digital combining (DC) versus the number of antennas $N$ when SNR is 0dB.  Here $b$ refers to the optimal number of quantizer bins. Total ADC power consumption is fixed at $20$mW. $W_{tot}=7$ Ghz,  $c=494$ fJ.}
\end{figure}

\begin{figure}[ht]\label{sim6}
\centering
\includegraphics[scale=0.65,trim= 0 0 0 0]{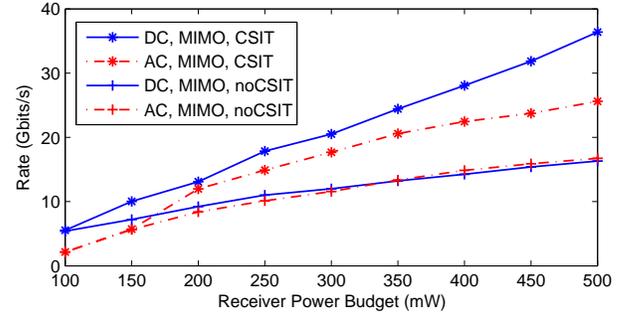}
\caption{Achievable rates  with analog (AC) and digital combining (DC)  versus total receiver power budget $P_{tot}$ when SNR is 0dB. Rates are computed under optimal number of antennas $N=N_t=N_r$, bandwidth and ADC resolution. $W_{tot}=7$ Ghz, $P_{LNA}=39$mW, $P_C=19.5$mW, $P_{PS}=19.5$mW, and $P_M=16.8$mW, $c=494$ fJ.}
\end{figure}

Finally, we consider all components of the front-end receiver, namely  the LNA, combiner, phase shifter, mixer and ADC, in computing the total receiver power consumption. We set $P_{LNA}=39$mW, $P_C=19.5$mW, $P_{PS}=19.5$mW \cite{Kramer}, and $P_M=16.8$mW \cite{yu}.  We optimize the rate of each MIMO architecture with analog and digital combining, with and without CSIT,  over the number of antennas $N=N_t=N_r$, bandwidth and ADC resolution for a given total receiver power budget $P_{tot}$.  We fix the SNR to be 0 dB. In Figure 6, we provide the achievable rates versus the total receiver power budget $P_{tot}$. In addition, Table I provides the maximum achievable rate in Gbits/s, the optimal resolution $b$, the optimal number of  antennas $N=N_t=N_r$, and bandwidth $W$ in Ghz for different values of the  total receiver power budget $P_{tot}$ under different combining and CSIT scenarios.

From Figure 6  we observe that digital combining with CSIT has higher rate than analog combining with CSIT, and the gap in between increases as the power budget increases.  When there is no CSIT, analog combining performs  better than the digital one only when the total power budget is greater than 350mW, and even in that case the performance improvement is very small. These observations are in contrast with Figure 4, where only the ADC power budget was considered. Comparing equations (\ref{pow1}) and (\ref{pow2}), we see that under a total receiver budget constraint, the power consumption of analog combining also starts increasing with the number of antennas, which helps emphasize the spatial multiplexing advantage of digital architectures. Table I  shows that  optimal resolution level of digital combining is always less than that of the analog one except for $P_{tot}=100$ mW. As expected, the optimal number of antennas increases with $P_{tot}$. However, we observe that the optimal bandwidth utilization may fluctuate as $P_{tot}$ increases. This is due to fact that in some cases it is better to spend the additional power to increase the number of antennas or the increase the resolution while reducing the bandwidth.

\section{Conclusions} \label{s:conc}

In this paper, we have studied a point-to-point MIMO mmWave communication system under two receiver architectures: Digital combining where each antenna element has a separate pair of ADCs, or analog combining where the signals are combined in analog domain so that only one ADC pair is required. For these receiver architectures we have investigated the effects of the bandwidth, resolution of ADCs and the number of antennas at the receiver on the achievable rate under the constraint of the total front-end receiver power or only the ADC power. Using AQNM, first, we have shown that there is an optimal bandwidth and resolution of ADC for the SISO channel. Then, we have derived achievable rates of the MIMO system with both analog combining and digital combining, with and without CSIT.  When  only ADC power budget is constrained, we have illustrated that depending on the operating regime analog combining may have higher rate but digital combining utilizes less bandwidth. However, when power consumption of all the receiver components are taken into account, we have shown that digital combining may have higher rate, and in some cases it can be optimal to increase the number of antennas or the increase the resolution while reducing the bandwidth.

\end{document}